\documentclass[10pt,conference,twocolumn]{IEEEtran}
\IEEEoverridecommandlockouts
\ifCLASSINFOpdf
\else
\fi
\usepackage{amsmath,amssymb,graphicx}
\usepackage{mathtools}
\usepackage{algpseudocode}
\usepackage{algorithm}
\usepackage{color}
\usepackage{bm}
\usepackage{amsfonts}
\usepackage{amscd}
\usepackage[mathcal]{eucal}
\usepackage{epsfig}
\usepackage{caption}
\usepackage{array}
\usepackage{eqparbox}
\usepackage{enumerate}
\usepackage{stackrel}
\usepackage{pdflscape}
\usepackage{lineno}
\usepackage{verbatim}

\newcommand{\etc}{{\em etc}}

\newcommand{\secref}[1]{Section~\ref{#1}}

\newcommand{\thrmref}[1]{Theorem~\ref{#1}}

\newtheorem{thrm}{\textbf{Theorem}}

\newcommand{\diag}{\mathrm{diag}}

\newcommand{\norm}[1]{\Vert#1\Vert}

\newenvironment{proof}[1][Proof]{\begin{trivlist}
\item[\hskip \labelsep {\bfseries #1}]}{\end{trivlist}}

\newcommand{\qed}{\nobreak \ifvmode \relax \else
      \ifdim\lastskip<1.5em \hskip-\lastskip
      \hskip1.5em plus0em minus0.5em \fi \nobreak
      \vrule height0.75em width0.5em depth0.25em\fi}

\begin{document}
\title{Reduced-rank Analysis of the Total Least Squares}
\author{
\authorblockN{Nagananda Kyatsandra}
\authorblockA{LTCI, T\'el\'ecom ParisTech,\\ Institut Mines - T\'el\'ecom,\\
Paris 75013, France.\\
Email: \texttt{nkyatsandra@enst.fr}}
\and
\authorblockN{Pramod P. Khargonekar}
\authorblockA{Department of EECS,\\ University of California,\\
Irvine, CA 92697, USA.\\
Email: \texttt{pramod.khargonekar@uci.edu}}
}

\maketitle
\thispagestyle{empty}
\pagestyle{empty}

\begin{abstract}
The reduced-rank method exploits the distortion-variance tradeoff to yield superior solutions for classic problems in statistical signal processing such as parameter estimation and filtering. The central idea is to reduce the variance of the solution at the expense of introducing a little distortion. In the context of parameter estimation, this yields an estimator whose sum of distortion plus variance is smaller than the variance of its undistorted counterpart. The method intrinsically results in an ordering mechanism for the singular vectors of the system matrix in the measurement model used for estimating the parameter of interest.  According to this ordering rule, only a few \emph{dominant} singular vectors need to be selected to construct the reduced-rank solution while the rest can be discarded. The reduced-rank estimator is less sensitive to measurement errors. In this paper, we attempt to derive the reduced-rank estimator for the total least squares (TLS) problem, including the order selection rule. It will be shown that, due to the inherent structure of the problem, it is not possible to exploit the distortion-variance tradeoff in TLS formulations using existing techniques, except in some special cases. This result has not been reported previously and warrants a new viewpoint to achieve rank reduction for the TLS estimation problem. The work is motivated by the problems arising in practical applications such as channel estimation in wireless communication systems and time synchronization in wireless sensor networks. 
\end{abstract}


\section{Introduction}\label{sec:introduction}
The underlying philosophy of rank reduction in statistical signal processing is to prioritize highly informative content of the measurement matrix in order to draw inferences that are superior to those obtained without reducing the rank of the matrix \cite[Chapter 9.9]{Scharf1991a}. The method brings into light the fundamental distortion-variance tradeoff which can be exploited to improve the performance of standard inference procedures. The definitions of distortion and variance depend on the problem at hand. For example, for parameter estimation using the linear least squares (LS), it has been shown that by introducing bias, rank reduction decreases the variance of the LS estimator \cite{Thorpe1989}, \cite{Cook1999}. However, the sum of bias plus variance is smaller than the variance of the unbiased estimator, thereby improving the overall mean-squared error performance.

Rank reduction for estimation entails prioritizing, or arranging, the products of each singular vector of the measurement matrix and the observation vector in the decreasing order of their magnitudes, and selecting only a few of these singular vectors to construct a reduced-rank estimator. As it happens, the order-selection rule emerges naturally when the bias-variance tradeoff is exploited. Eliminating a few singular vectors reduces the complexity of the prior model leading to a lower mean-squared error of the estimator, and at the same time decreases the sensitivity to measurement errors. It has to be emphasized that, though rank reduction involves singular value decomposition (SVD) of the system matrix, the computational issue is not the topic of concern; the main goal is to seek an estimator with improved overall performance.

Reduced-rank LS estimator of \cite{Thorpe1989} assumes complete knowledge of the measurement matrix.  In many applications ranging from channel estimation in wireless communication systems and sensor networks, automatic control, finance, computational biology, {\etc}, it is generally the case that not only are the observations noisy, but the elements of the measurement matrix are also corrupted by noise. This type of modeling is referred to as the total least squares (TLS), and is a powerful extension of the LS idea which corresponds to partial modification of the data \cite{Huffel1991}. In the statistical community, TLS is referred to as errors-in-variables (EIV) modeling, orthogonal regression \cite{Huffel2002}, or forward stepwise regression \cite{Cook1999}. 

To the best of our knowledge, the reduced-rank analysis of the problem of parameter estimation using the TLS formulation has not been reported in the literature, and is the subject of this paper. The main result of the paper is given in  \thrmref{thrm:reduced_tls}, which essentially states that, owing to the inherent structure of the formulation, it is in general not possible to exploit the bias-variance tradeoff for the TLS problem. Hence, the classical approach of deriving the reduced-rank estimator to the LS problem does not yield a reduced-rank estimator to the TLS setup. However, for some special cases (for example, when the unknown parameter to be estimated is norm-constrained), a reduced-rank estimator does exist. For this special case, the corresponding order selection rule is similar in spirit to the one derived for the reduced-rank LS estimator with some minor technical differences which will be highlighted.

In \secref{sec:reduced_ls}, we provide references to existing literature in the area of rank reduction and re-derive the reduced-rank LS estimator. In \secref{sec:reduced_tls}, we show that it is not possible to exploit the bias-variance tradeoff for the TLS formulation. Some remarks can be found in \secref{sec:conclusion}. Notation: Vectors and matrices are denoted by bold lower and upper case letters, respectively. The expectation operator is denoted by $\mathbb{E}[\cdot]$, while $\norm{\cdot}$ denotes the Euclidean norm. The transpose of a matrix $\bm{M}$ is denoted by $\bm{M}^{\mathrm{T}}$. The central Chi square distribution with $k$ degrees of freedom is denoted by $\chi^2_{k}$. 

\section{Existing results and the reduced-rank LS estimator}\label{sec:reduced_ls}
\subsection{Literature review}
Reduced-rank solutions to several standard problems in signal processing have appeared in the literature. A few examples include reduced-rank solutions to stationary time series modeling, stationary time series whitening, and vector quantization \cite{Scharf1987}. Rank reduction leads to a decreased complexity of decoding linear block codes in a complex field \cite{Scharf1987a}. A comprehensive coverage of the reduced-rank processing can be found in \cite{Scharf1991}, where several linear signal models including the LS, Wiener filters, minimum variance distortionless look beamformer, block quantizers, Gauss-Gauss detectors have been considered. The bias-variance tradeoff for low-rank estimators of higher order statistics using a tensor product formulation for the moments and cumulants is analyzed in \cite{Andre1997}. Novel rank selection schemes for adaptive filtering and space-time adaptive processing are presented in \cite{Goldstein1997} and \cite{Guerci2000}. 

Subspace-based methods for computing the optimal reduced-rank estimators and filters are analyzed in \cite{Hua2001}. Optimal rank selection for multistage Wiener filter is presented in \cite{Hiemstra2002}. Rank reduction for complex random vectors and wide-sense stationary signals using principal components is derived in \cite{Schreier2003}. Adaptive reduced-rank processing for communication systems can be found in \cite{Honig2001} - \nocite{Nicoli2005}\nocite{deLamare2009}\nocite{deLamare2010}\cite{deLamare2011}. Maximum likelihood estimation for reduced-rank linear regression is presented in \cite{Stoica1996}. The low-rank approximation of a matrix by one of the same dimension but smaller rank is derived in \cite{Manton2003}. Finally, it is of interest to note that the ordering principle emerging out of rank reduction has inspired several important directions in statistical signal processing. For example, ordering observations opened up a new perspective on finite sample analysis of signal detection, classification and estimation, where optimal performance was achieved using fewer than half the number of available samples \cite{Blum2008} - \nocite{Blum2011}\cite{Sriranga2018a}. 

The main idea of rank reduction in the context of statistical signal processing was reported in the seminal work of Thorpe and Scharf for parameter estimation using LS in \cite{Thorpe1989}, and is reproduced below. 

\subsection{Reduced-rank LS estimator}\label{subsec:reduced_ls}
In this subsection, we re-derive in significant detail the reduced-rank LS estimator developed in \cite{Thorpe1989}, and set the stage for the TLS model considered in \secref{sec:reduced_tls}. We begin with the following linear model:
\begin{eqnarray}
\begin{split}
\bm{y} &= \bm{x} + \bm{n} \\
&= \bm{H}\bm{\theta} + \bm{n},
\end{split}
\label{eq:linear_model}
\end{eqnarray}
where $\bm{y}$, $\bm{x}$ and $\bm{n}$ are $N \times 1$ dimensional vectors, $\bm{H}$ is a known $N \times p$ dimensional matrix and $\bm{\theta}$ is a $p \times 1$ dimensional unknown vector. When the $p$ columns of $\bm{H}$ are linearly independent, then there will be only $(N-p)$ linearly independent vectors that can be orthogonal to the $p$ columns of $\bm{H}$. The noise vector $\bm{n} \sim \mathcal{N}(\bm{0}_N, \sigma^2\bm{I}_N)$, where $\bm{0}_N$ and $\bm{I}_N$ denote the $N \times 1$ vector of all zeros and the $N \times N$ identity matrix, respectively, while $\sigma^2$ denotes the noise variance. Therefore, $\bm{y} \sim \mathcal{N}(\bm{H}\bm{\theta}, \sigma^2\bm{I}_N)$. The least squares estimates of the parameter $\bm\theta$, the signal $\bm{x}$ and noise $\bm{n}$ are denoted by $\hat{\bm\theta}_{\text{LS}}$, $\hat{\bm{x}}_{\text{LS}}$ and $\hat{\bm{n}}_{\text{LS}}$, respectively. These are given by 
\begin{eqnarray}
\hat{\bm\theta}_{\text{LS}} &=& (\bm{H}^{\mathrm{T}}\bm{H})^{-1}\bm{H}^{\mathrm{T}}\bm{y}, \label{eq:ls_theta} \\
\hat{\bm{x}}_{\text{LS}} &=& \bm{H}(\bm{H}^{\mathrm{T}}\bm{H})^{-1}\bm{H}^{\mathrm{T}}\bm{y}, \label{eq:ls_x} \\
\nonumber \hat{\bm{n}}_{\text{LS}} &=& \bm{y} - \hat{\bm{x}} \\
&=& [\bm{I}_N -  \bm{H}(\bm{H}^{\mathrm{T}}\bm{H} )^{-1}\bm{H}^{\mathrm{T}}]\bm{y}, \label{eq:ls_n}
\end{eqnarray}
and are unique if the inverse of $\bm{H}^{\mathrm{T}}\bm{H}$ exists. The estimators $\hat{\bm\theta}_{\text{LS}}$, $\hat{\bm{x}}_{\text{LS}}$ and $\hat{\bm{n}}_{\text{LS}}$ are linear transformations on the multivariate normal random vector $\bm{y}$, governed by the following distributions:
\begin{eqnarray}
\hat{\bm\theta}_{\text{LS}} &\sim& \mathcal{N}\left(\bm{\theta}, \sigma^2(\bm{H}^{\mathrm{T}}\bm{H})^{-1}\right), \label{eq:ls_theta_dist} \\
\hat{\bm{x}}_{\text{LS}} &\sim & \mathcal{N}\left(\bm{H}\bm{\theta}, \sigma^2\bm{H}(\bm{H}^{\mathrm{T}}\bm{H})^{-1}\bm{H}^{\mathrm{T}}\right), \label{eq:ls_x_dist} \\
\hat{\bm{n}}_{\text{LS}} &\sim & \mathcal{N}\left(\bm{0}_N, \sigma^2[\bm{I}_N -  \bm{H}(\bm{H}^{\mathrm{T}}\bm{H})^{-1}\bm{H}^{\mathrm{T}}]\right) \label{eq:ls_n_dist}.
\end{eqnarray}
From \eqref{eq:ls_x_dist} it can be seen that $\hat{\bm{x}}_{\text{LS}}$ is an unbiased estimator of $\bm{x}$. The squared error $\hat{\bm{x}}_{\text{LS}}^{\mathrm{T}}\hat{\bm{x}}_{\text{LS}} \sim \sigma^2 \chi^2_{p}$, and the mean squared error of the estimator $\hat{\bm{x}}_{\text{LS}}$ is $\mathbb{E} [(\hat{\bm{x}}_{\text{LS}} - \bm{x})^2] = p\sigma^2$. 

In reduced-rank processing, we will seek to achieve a smaller mean-squared error than $p\sigma^2$ albeit at the price of nonzero mean. This corresponds to a bias-variance tradeoff. Towards this end, we first express \eqref{eq:ls_theta_dist} - \eqref{eq:ls_n_dist} in terms of the SVD of the matrix $\bm{H}$ which is given by $\bm{H} = \bm{U}_H \bm{\Gamma}_H \bm{V}^{\mathrm{T}}_H$, where $\bm{U}_H = \left[\bm{u}_1, \dots, \bm{u}_p \right] \in \mathbb{R}^{N \times p}$ and $\bm{V}_H = \left[\bm{v}_1, \dots, \bm{v}_p \right] \in \mathbb{R}^{p \times p}$ are orthogonal matrices, and $\bm{\Gamma}_H = \diag \left[\gamma_1, \dots, \gamma_p \right] \in \mathbb{R}^{p \times p}$ is a diagonal matrix comprising the singular values $\gamma_1 \geq \dots \geq \gamma_p$. Thus, \eqref{eq:ls_theta_dist} - \eqref{eq:ls_n_dist}  can be written as 
\begin{eqnarray}
\hat{\bm\theta}_{\text{LS}} &\sim& \mathcal{N}\left(\bm{\theta}, \sigma^2\bm{V}_H\bm{\Gamma}^{-2}_H\bm{V}^{\mathrm{T}}_H\right), \label{eq:ls_x_dist_svd} \\
\hat{\bm{x}}_{\text{LS}} &\sim & \mathcal{N}\left(\bm{H}\bm{\theta}, \sigma^2\bm{U}_H\bm{U}^{\mathrm{T}}_H\right), \label{eq:ls_x_dist_svd} \\
\hat{\bm{n}}_{\text{LS}} &\sim & \mathcal{N}\left(\bm{0}_N, \sigma^2\left[\bm{I}_N - \bm{U}_H\bm{U}^{\mathrm{T}}_H\right]\right). \label{eq:ls_n_dist_svd}
\end{eqnarray}

The full-rank estimator $\hat{\bm{x}}_{\text{LS}}$ will be replaced by a low-rank estimator 
\begin{eqnarray}
\hat{\bm{x}}_{r, {\text{LS}}} \triangleq \bm{U}_r\bm{U}^{\mathrm{T}}_r\bm{y},
\label{eq:reduced_ls}
\end{eqnarray}
where the $N \times r$ matrix $\bm{U}_r = [\bm{u}_{(1)}, \dots, \bm{u}_{(r)}]$ is obtained by discarding $(p - r)$ orthogonal vectors that comprise $\bm{U}_H$ and $\bm{u}_{(j)}$ denotes the $j^{\text{th}}$ ``ordered'' orthogonal vector which is not necessarily the $j^{\text{th}}$ vector. The notion of ordering will become clearer as we proceed. The estimation error $(\bm{x} - \hat{\bm{x}}_{r, {\text{LS}}}) \sim \mathcal{N}(\bm{x} - \bm{x}_r, \sigma^2\bm{U}_r\bm{U}^{\mathrm{T}}_r)$, where $(\bm{x} - \bm{x}_r) = \bm{b}_{\text{LS}}$ denotes the bias of the LS estimator. The mean-squared error of the reduced-rank estimator $\hat{\bm{x}}_{r, {\text{LS}}}$ is given by
\begin{eqnarray}
\nonumber \text{mse}(r)\!\!\! &=&\!\!\! \mathbb{E}\left\{[\bm{x} - \hat{\bm{x}}_{r,\text{LS}}]^{\mathrm{T}}[\bm{x} - \hat{\bm{x}}_{r,\text{LS}}] \right\} \\
\nonumber  &=&\!\!\! \mathbb{E}\left\{[\bm{x}^{\mathrm{T}}\bm{x} - \bm{x}^{\mathrm{T}}\hat{\bm{x}}_{r,\text{LS}} - \hat{\bm{x}}_{r,\text{LS}}^{\mathrm{T}}\bm{x} - \hat{\bm{x}}_{r,\text{LS}}^{\mathrm{T}}\hat{\bm{x}}_{r,\text{LS}}] \right\} \\
\nonumber  &=&\!\!\! \mathbb{E}\left\{\bm{x}^{\mathrm{T}}\bm{x}\right\} - \mathbb{E}\left\{\bm{x}^{\mathrm{T}}\hat{\bm{x}}_{r,\text{LS}} \right\} - \mathbb{E}\left\{\hat{\bm{x}}_{r,\text{LS}}^{\mathrm{T}}\bm{x} \right\} \\ \nonumber && +  \mathbb{E}\left\{\hat{\bm{x}}_{r,\text{LS}}^{\mathrm{T}}\hat{\bm{x}}_{r,\text{LS}} \right\} \\
\nonumber  &=&\!\!\! \bm{x}^{\mathrm{T}}\bm{x} - \bm{x}^{\mathrm{T}}\bm{x}_r - \bm{x}_r^{\mathrm{T}}\bm{x}  +  \mathbb{E}\left\{\hat{\bm{x}}_{r,\text{LS}}^{\mathrm{T}}\hat{\bm{x}}_{r,\text{LS}} \right\}\\
\nonumber  &=&\!\!\! \bm{x}^{\mathrm{T}}\bm{x} - \bm{x}^{\mathrm{T}}\bm{x}_r - \bm{x}_r^{\mathrm{T}}\bm{x}  +  \mathbb{E}\left\{\bm{y}^{\mathrm{T}}\bm{U}_r\bm{U}^{\mathrm{T}}_r\bm{U}_r\bm{U}^{\mathrm{T}}_r\bm{y}\right\}\\
\nonumber  &=&\!\!\! \bm{x}^{\mathrm{T}}\bm{x} - \bm{x}^{\mathrm{T}}\bm{x}_r - \bm{x}_r^{\mathrm{T}}\bm{x} \\ \nonumber  && +  \mathbb{E}\left\{(\bm{x} + \bm{n})^{\mathrm{T}}\bm{U}_r\bm{U}^{\mathrm{T}}_r(\bm{x} + \bm{n})\right\}\\
\nonumber  &=&\!\!\! \bm{x}^{\mathrm{T}}\bm{U}_H\bm{U}^{\mathrm{T}}_H\bm{x} - \bm{x}^{\mathrm{T}}\bm{U}_r\bm{U}^{\mathrm{T}}_r\bm{x} - \bm{x}^{\mathrm{T}}\bm{U}_r\bm{U}^{\mathrm{T}}_r\bm{x} \\ \nonumber  &&+ \mathbb{E}\left\{\bm{x}^{\mathrm{T}} \bm{U}_r\bm{U}^{\mathrm{T}}_r\bm{x} +  \bm{x}^{\mathrm{T}} \bm{U}_r\bm{U}^{\mathrm{T}}_r\bm{n}\right.\\ \nonumber  && \left. + \bm{n}^{\mathrm{T}} \bm{U}_r\bm{U}^{\mathrm{T}}_r\bm{x}  + \bm{n}^{\mathrm{T}} \bm{U}_r\bm{U}^{\mathrm{T}}_r\bm{n}\right\} \\
\nonumber  &=&\!\!\! \bm{x}^{\mathrm{T}}\bm{U}_H\bm{U}^{\mathrm{T}}_H\bm{x} - \bm{x}^{\mathrm{T}}\bm{U}_r\bm{U}^{\mathrm{T}}_r\bm{x} - \bm{x}^{\mathrm{T}}\bm{U}_r\bm{U}^{\mathrm{T}}_r\bm{x} \\ \nonumber  && + \bm{x}^{\mathrm{T}} \bm{U}_r\bm{U}^{\mathrm{T}}_r\bm{x}  + \mathbb{E}\left\{\bm{n}^{\mathrm{T}} \bm{U}_r\bm{U}^{\mathrm{T}}_r\bm{n}\right\} \\
\nonumber  &=&\!\!\! \bm{x}^{\mathrm{T}}\bm{U}_H\bm{U}^{\mathrm{T}}_H\bm{x} - \bm{x}^{\mathrm{T}}\bm{U}_r\bm{U}^{\mathrm{T}}_r\bm{x} + \mathbb{E}\left\{\bm{n}^{\mathrm{T}} \bm{U}_r\bm{U}^{\mathrm{T}}_r\bm{n}\right\} \\
\nonumber  &=&\!\!\!\!\! \bm{x}^{\mathrm{T}}\bm{U}_H\bm{U}^{\mathrm{T}}_H\bm{x} - \bm{x}^{\mathrm{T}}\bm{U}_r\bm{U}^{\mathrm{T}}_r\bm{x} + \text{tr}\left[\mathbb{E}\left\{\bm{n}\bm{U}_r\bm{U}^{\mathrm{T}}_r\bm{n}^{\mathrm{T}}\right\}\right] \\
&=&\!\!\! \sum_{j = r+1}^{p}\norm{\bm{u}^{\mathrm{T}}_{(j)}\bm{x}}^2 + r \sigma^2.
\label{eq:mse_r}
\end{eqnarray}

\begin{thrm} 
Consider the LS model in \eqref{eq:linear_model} and the reduced-rank LS estimator given by \eqref{eq:reduced_ls}. There exists an $r = r^{\ast}$ given by 
\begin{eqnarray}
r^{\ast} = \arg\min_{r} \left[\sum_{j = r+1}^{p}\norm{\bm{u}^{\mathrm{T}}_{(j)}\bm{y}}^2 + \sigma^2(2r - p)\right],
\end{eqnarray}
that minimizes the mean squared error between $\bm{x}$ and $\hat{\bm{x}}_{r, {\text{LS}}}$.
\label{thrm:reduced_ls}
\end{thrm}
\begin{proof}
The estimator $\hat{\bm{x}}_{r, {\text{LS}}} \sim \mathcal{N}(\bm{x}_r, \sigma^2\bm{U}_r\bm{U}^{\mathrm{T}}_r)$, where $\bm{x}_r = \bm{U}_r\bm{U}^{\mathrm{T}}_r\bm{x}$ is the projection of $\bm{x}$ onto the span of $\bm{U}_r\bm{U}^{\mathrm{T}}_r$. The rank reduction procedure will reduce the variance of the estimator of $\bm{x}$ whenever 
\begin{eqnarray}
p\sigma^2  > \sum_{j = r+1}^{p}\norm{\bm{u}^{\mathrm{T}}_{(j)}\bm{x}}^2 + r \sigma^2,
\label{eq:reduce__ls_when}
\end{eqnarray}
which suggests that the optimum choice of the rank $r$ is 
\begin{eqnarray}
\begin{split}
r^{\ast} &= \arg\min_{r} \text{mse}(r) \\
&= \arg\min_{r} \sum_{j = r+1}^{p}\norm{\bm{u}^{\mathrm{T}}_{(j)}\bm{x}}^2 + r \sigma^2.
\end{split}
\label{eq:optimum_r}
\end{eqnarray}
However, since the signal vector $\bm{x}$ is unknown, we replace $\text{mse}(r)$ with its estimate to solve the setup in \eqref{eq:optimum_r}. Towards this end, we first estimate the bias $\bm{b}_{\text{LS}}$ using the following statistic:
\begin{eqnarray}
\begin{split}
\hat{\bm{b}}_{\text{LS}}  &= \left(\bm{U}_H\bm{U}^{\mathrm{T}}_H - \bm{U}_r\bm{U}^{\mathrm{T}}_r\right)\bm{y}  \\ &
\sim \mathcal{N}\left(\bm{b}_{\text{LS}}, \sigma^2(\bm{U}_H\bm{U}^{\mathrm{T}}_H - \bm{U}_r\bm{U}^{\mathrm{T}}_r) \right).
\end{split}
\label{eq:est_bias_ls}
\end{eqnarray}
The mean-squared error of the estimator $\hat{\bm{b}}_{\text{LS}}$ is given by 
\begin{eqnarray}
\nonumber \mathbb{E}\{[\hat{\bm{b}}_{\text{LS}} - \bm{b}_{\text{LS}}]^{\mathrm{T}}[\hat{\bm{b}}_{\text{LS}} - \bm{b}_{\text{LS}}]\}\!\!\!\!\!\!\!  &=& \!\!\!\!\!  \mathbb{E}\{\hat{\bm{b}}_{\text{LS}}^{\mathrm{T}}\hat{\bm{b}}_{\text{LS}} -  \hat{\bm{b}}_{\text{LS}}^{\mathrm{T}}\bm{b}_{\text{LS}}  - \bm{b}^{\mathrm{T}}_{\text{LS}}\hat{\bm{b}}_{\text{LS}}\\ \nonumber && + \bm{b}^{\mathrm{T}}_{\text{LS}}\bm{b}_{\text{LS}}\} \\
\nonumber &\stackrel{(i)}=& \!\!\!\!\! \mathbb{E}\{\hat{\bm{b}}_{\text{LS}}^{\mathrm{T}}\hat{\bm{b}}_{\text{LS}} \} - \bm{b}^{\mathrm{T}}_{\text{LS}}\bm{b}_{\text{LS}} \\
\nonumber &\stackrel{(ii)}=& \!\!\!\!\! \text{tr}\left[\sigma^2(\bm{U}_H\bm{U}^{\mathrm{T}}_H - \bm{U}_r\bm{U}^{\mathrm{T}}_r) \right] \\
&=& \!\!\!\!\! \sigma^2(p-r),
\label{eq:mse_b}
\end{eqnarray}
where $(i)$ and $(ii)$ follow from \eqref{eq:est_bias_ls}. From \eqref{eq:mse_b}, we see that 
\begin{eqnarray}
\mathbb{E}\{\hat{\bm{b}}_{\text{LS}}^{\mathrm{T}}\hat{\bm{b}}_{\text{LS}} \} =  \sigma^2(p-r) + \bm{b}^{\mathrm{T}}_{\text{LS}}\bm{b}_{\text{LS}}
\end{eqnarray}
which implies that the estimator $\hat{\bm{b}}_{\text{LS}}^{\mathrm{T}}\hat{\bm{b}}_{\text{LS}}$ must be corrected by $-\sigma^2(p-r)$ to be an unbiased estimator of $\bm{b}^{\mathrm{T}}_{\text{LS}}\bm{b}_{\text{LS}}$. This leads to the following estimator for $\text{mse}(r)$:
\begin{eqnarray}
\begin{split}
\hat{\text{mse}}(r) &= \hat{\bm{b}}_{\text{LS}}^{\mathrm{T}}\hat{\bm{b}}_{\text{LS}} - \sigma^2(p-r) + r\sigma^2 \\
&= \hat{\bm{b}}_{\text{LS}}^{\mathrm{T}}\hat{\bm{b}}_{\text{LS}} + \sigma^2(2r - p).
\end{split}
\label{eq:est_mse_r}
\end{eqnarray}
The optimum choice of $r$ is, therefore, given by 
\begin{eqnarray}
\begin{split}
r^{\ast} &= \arg\min_{r} \hat{\text{mse}}(r) \\
&\stackrel{(iii)}= \arg\min_{r} \left[\sum_{j = r+1}^{p}\norm{\bm{u}^{\mathrm{T}}_{(j)}\bm{y}}^2 + \sigma^2(2r - p)\right],
\end{split}
\label{eq:optimum_r_mod}
\end{eqnarray}
where $(iii)$ follows from \eqref{eq:est_bias_ls}. This completes the proof of \thrmref{thrm:reduced_ls}.
\end{proof}

From \eqref{eq:optimum_r_mod} it is clear that the eigenvectors of $\bm{U}_H$ should be ordered such that 
\begin{eqnarray}
\norm{\bm{u}^{\mathrm{T}}_{(1)}\bm{y}}^2 \geq \cdots \geq \norm{\bm{u}^{\mathrm{T}}_{(r)}\bm{y}}^2 \geq \cdots \geq \norm{\bm{u}^{\mathrm{T}}_{(p)}\bm{y}}^2,
\label{eq:order_ls}
\end{eqnarray}
and the dominant $r$ eigenvectors should be used to construct the rank$-r$ projector $\bm{U}_r\bm{U}^{\mathrm{T}}_r$. Thus, as we remarked in \secref{sec:introduction}, ordering emerges naturally when the bias-variance tradeoff is exploited for the LS estimator. 

Solving \eqref{eq:optimum_r_mod} numerically provides the optimal $r^{\ast}$. The reduced-rank estimator is obtained by conveniently discarding $(p - r^{\ast})$ columns in the matrix $\bm{U}_H = \left[\bm{u}_1, \dots, \bm{u}_p \right] \in \mathbb{R}^{N \times p}$, and is given by $\hat{\bm{x}}_{r^{\ast}, {\text{LS}}} \triangleq \bm{U}_{r^{\ast}}\bm{U}^{\mathrm{T}}_{r^{\ast}}\bm{y}$, whose sum of bias plus variance is smaller than the variance of the unbiased estimator $\hat{\bm{x}}_{\text{LS}} = \bm{H}(\bm{H}^{\mathrm{T}}\bm{H})^{-1}\bm{H}^{\mathrm{T}}\bm{y}$ given by \eqref{eq:ls_x}. 

The reduced-rank estimator $\hat{\bm{x}}_{r, {\text{LS}}}$ has a lower variance at the expense of a higher bias compared to its full-rank counterpart $\hat{\bm{x}}_{\text{LS}}$. However, complete knowledge of the measurement matrix $\bm{H}$ is assumed which is unreasonable. In fact, in many practical situations, this assumption is invalid. For example, in channel estimation for wireless communication networks, the network should be observed for a long period of time to obtain an estimate of the system matrix which could be practically infeasible \cite[Chapter 9.1]{Kim2015}. Similarly, delays in obtaining an accurate estimate of the system matrix could severely impact time synchronization for consensus in wireless sensor networks \cite{Serpedin2009}. Therefore, assuming perfect knowledge of $\bm{H}$ could be inaccurate. Thus, one is restricted to work on the noisy version of matrix $\bm{H}$. In the next section, we consider probabilistic knowledge of $\bm{H}$ leading to the TLS model and show why it does not permit a reduced-rank solution. 

\section{Reduced-rank analysis of the TLS problem}\label{sec:reduced_tls}
In the TLS problem, the entries of the system matrix $\bm{H}$ considered in the linear model \eqref{eq:linear_model} are corrupted by noise leading to the following observations model:
\begin{eqnarray}
\begin{split}
\bm{y} &= \bm{x} + \bm{n} = \bm{H}\bm{\theta} + \bm{n}, \\
\tilde{\bm{H}} &= \bm{H} + \bm{E},
\end{split}
\label{eq:tls}
\end{eqnarray}
where the rows of the $N \times p$ matrix $\bm{E}$ of errors are sampled from the distribution $\mathcal{N}(\bm{0}_p, \sigma^2\bm{I}_p)$, where $\bm{0}_p$ and $\bm{I}_p$ denote the $p \times 1$ vector of all zeros and the $p \times p$ identity matrix, respectively. The distribution of $\bm{y}$ is given by $\bm{y} \sim \mathcal{N}(\bm{H}\bm{\theta}, \sigma^2 [1 + \bm{\theta}^{\mathrm{T}}\bm{\theta}]\bm{I}_N)$. Given $\bm{y}$ and $\tilde{\bm{H}}$, the TLS minimizes a sum of the squared normalized residuals expressed as follows \cite{Huffel1991}:
\begin{eqnarray}
\hat{\bm{\theta}}_{\text{TLS}} = \arg\min_{\bm{\theta}} \frac{\norm{\tilde{\bm{H}}\bm{\theta} - \bm{y}}^2}{\norm{\bm{\theta}}^2 + 1}.
\label{eq:tls_min}
\end{eqnarray}

Let us hypothesize the existence of a low rank estimator of the form
\begin{eqnarray}
\hat{\bm{x}}_{q, {\text{TLS}}} \triangleq \bm{U}_q\bm{U}^{\mathrm{T}}_q\bm{y},
\label{eq:reduced_tls}
\end{eqnarray}
where the $N \times q$ matrix $\bm{U}_q = [\bm{u}_{(1)}, \dots, \bm{u}_{(q)}]$ is obtained by discarding $(p - q)$ orthogonal vectors that comprise $\bm{U}_s$, and $\bm{u}_{(j)}$ denotes the $j^{\text{th}}$ ordered orthogonal vector which is not necessarily the $j^{\text{th}}$ vector. The notion of ordering is similar to the one introduced in \secref{subsec:reduced_ls}; however, it will soon become evident that such a low rank estimator does not exist (except under certain conditions) and that the ordering mechanism similar to the one shown in \eqref{eq:order_ls} cannot be realized for the TLS problem. 

\begin{thrm} 
Consider the TLS formulation given by \eqref{eq:tls}, and hypothesize that there exists a low rank estimator $\hat{\bm{x}}_{q, {\text{TLS}}}$ given by \eqref{eq:reduced_tls}. There exists no $q < p$ (except in some special cases which will be explained in \secref{sec:conclusion}) independent of the unknown parameter $\bm{\theta}$ that minimizes the mean squared error between $\bm{x}$ and $\hat{\bm{x}}_{q, {\text{TLS}}}$. Thus, the hypothesis of existence of a low rank estimator is false.
\label{thrm:reduced_tls}
\end{thrm}

\begin{proof}
For clarity of exposition, we will proceed along the lines of the LS problem considered in \secref{sec:reduced_ls}. To obtain the TLS estimates of the parameter $\bm{\theta}$, the signal $\bm{x}$ and the noise $\bm{n}$, we first obtain the SVD of the $N \times (p+1)$ augmented matrix $\bm{A} \triangleq [\tilde{\bm{H}}~~\bm{y}]$ given by $\bm{A} = \bm{U}_A \Sigma_A \bm{V}_A^{\mathrm{T}}$, where $\bm{U}_A = \left[\bm{u}_{A,1}, \dots, \bm{u}_{A,p+1} \right] \in \mathbb{R}^{N \times (p+1)}$, $\bm{V}_A = \left[\bm{v}_{A,1}, \dots, \bm{v}_{A,p+1} \right] \in \mathbb{R}^{p+1 \times p+1}$
are orthogonal matrices, and $\bm{\Sigma}_A = \diag \left[\gamma_{A,1}, \dots, \gamma_{A,p+1} \right] \in \mathbb{R}^{p+1 \times p+1}$
is a diagonal matrix comprising the singular values $\gamma_{A,1} \geq \dots \geq \gamma_{A,p+1}$. 

The orthogonal matrix $\bm{U}_A$ can be written as $\bm{U}_A = [\bm{U}_s~~ \bm{u}_s]$, where $\bm{u}_s$ is the singular vector corresponding to the smallest singular value of $\bm{A}$. The TLS estimates of $\bm{\theta}$, $\bm{x}$ and $\bm{n}$ are expressed in terms of $\bm{U}_s$ as follows:
\begin{eqnarray}
\hat{\bm{\theta}}_{\text{TLS}} &=& (\tilde{\bm{H}}^{\mathrm{T}}\bm{U}_s\bm{U}^{\mathrm{T}}_s\tilde{\bm{H}})^{-1}\tilde{\bm{H}}^{\mathrm{T}}\bm{U}_s\bm{U}^{\mathrm{T}}_s\bm{y},  \label{eq:tls_theta} \\
\hat{\bm{x}}_{\text{TLS}} &=& \bm{U}_s\bm{U}^{\mathrm{T}}_s\bm{y}, \label{eq:tls_x} \\
\hat{\bm{n}}_{\text{TLS}} &=& (\bm{I} - \bm{U}_s\bm{U}^{\mathrm{T}}_s)\bm{y}. \label{eq:tls_n}
\end{eqnarray}
The estimated (or, corrected) system matrix is given by $\hat{\bm{H}} = \bm{U}_s\bm{U}^{\mathrm{T}}_s\tilde{\bm{H}}$. The estimator $\hat{\bm{x}}_{\text{TLS}} \sim \mathcal{N} (\hat{\bm{H}} \bm{\theta},  \sigma^2 [1 + \bm{\theta}^{\mathrm{T}}\bm{\theta}]\bm{U}_s\bm{U}^{\mathrm{T}}_s)$. That is, $\hat{\bm{x}}_{\text{TLS}}$ is a biased estimator of $\bm{x}$ unlike $\hat{\bm{x}}_{\text{LS}}$ [see \eqref{eq:ls_x_dist}]. The mean squared error of the TLS estimate is 
\begin{eqnarray}
[\norm{\hat{\bm{H}} \bm{\theta} - \bm{H} \bm{\theta}}^2 + \sigma^2 (1 + \bm{\theta}^{\mathrm{T}}\bm{\theta})p].
\end{eqnarray}   

For notational convenience, we let $\bm{x}_q = \bm{U}_q\bm{U}^{\mathrm{T}}_q\tilde{\bm{H}}\bm{\theta}$. It is seen that $\bm{x}_q \sim \mathcal{N}(\bm{U}_q\bm{U}^{\mathrm{T}}_q\bm{H}\bm{\theta}, \sigma^2\bm{\theta}^{\mathrm{T}}\bm{\theta}\bm{U}_q\bm{U}^{\mathrm{T}}_q)$, while the estimator $\hat{\bm{x}}_{q, {\text{TLS}}} \sim \mathcal{N}(\bm{x}_q, \sigma^2 [1 + \bm{\theta}^{\mathrm{T}}\bm{\theta}]\bm{U}_q\bm{U}^{\mathrm{T}}_q\bm{I}_N\bm{U}_q\bm{U}^{\mathrm{T}}_q)$ is a compound distribution which simplifies to 
\begin{align}
\hat{\bm{x}}_{q, {\text{TLS}}} \sim \mathcal{N}(\bm{U}_q\bm{U}^{\mathrm{T}}_q\bm{H}\bm{\theta}, \sigma^2 [1 + \bm{\theta}^{\mathrm{T}}\bm{\theta}]\bm{U}_q\bm{U}^{\mathrm{T}}_q \!\! +\!\! \sigma^2\bm{\theta}^{\mathrm{T}}\bm{\theta}\bm{U}_q\bm{U}^{\mathrm{T}}_q).
\end{align}

The mean-squared error of the reduced-rank estimator $\hat{\bm{x}}_{q, {\text{TLS}}}$ is given by
\begin{eqnarray}
\nonumber \text{mse}(q)\!\!\! &=&\!\!\! \mathbb{E}\left\{[\bm{x} - \hat{\bm{x}}_{q,\text{TLS}}]^{\mathrm{T}}[\bm{x} - \hat{\bm{x}}_{q,\text{TLS}}] \right\} \\
\nonumber  &=&\!\!\! \mathbb{E}\left\{[\bm{x}^{\mathrm{T}}\bm{x} - \bm{x}^{\mathrm{T}}\hat{\bm{x}}_{q,\text{TLS}} - \hat{\bm{x}}_{q,\text{TLS}}^{\mathrm{T}}\bm{x} + \hat{\bm{x}}_{q,\text{TLS}}^{\mathrm{T}}\hat{\bm{x}}_{q,\text{TLS}}] \right\} \\
\nonumber  &=&\!\!\! \mathbb{E}\left\{\bm{x}^{\mathrm{T}}\bm{x}\right\} - \mathbb{E}\left\{\bm{x}^{\mathrm{T}}\hat{\bm{x}}_{q,\text{TLS}} \right\} - \mathbb{E}\left\{\hat{\bm{x}}_{q,\text{TLS}}^{\mathrm{T}}\bm{x} \right\}\\ \nonumber && +  \mathbb{E}\left\{\hat{\bm{x}}_{q,\text{TLS}}^{\mathrm{T}}\hat{\bm{x}}_{q,\text{TLS}} \right\} \\
\nonumber  &=&\!\!\! \bm{x}^{\mathrm{T}}\bm{x} - \bm{x}^{\mathrm{T}}\bm{U}_q\bm{U}^{\mathrm{T}}_q\bm{x} - \bm{x}^{\mathrm{T}}\bm{U}_q\bm{U}^{\mathrm{T}}_q\bm{x} \\ \nonumber && +  \mathbb{E}\left\{\hat{\bm{x}}_{q,\text{TLS}}^{\mathrm{T}}\hat{\bm{x}}_{q,\text{TLS}} \right\}\\
\nonumber  &=&\!\!\! \bm{x}^{\mathrm{T}}\bm{x} - \bm{x}^{\mathrm{T}}\bm{U}_q\bm{U}^{\mathrm{T}}_q\bm{x} - \bm{x}^{\mathrm{T}}\bm{U}_q\bm{U}^{\mathrm{T}}_q\bm{x} \\ \nonumber && + \mathbb{E}\left\{\bm{y}^{\mathrm{T}}\bm{U}_q\bm{U}^{\mathrm{T}}_q\bm{y}\right\}\\
\nonumber  &=&\!\!\! \bm{x}^{\mathrm{T}}\bm{x}  - \bm{x}^{\mathrm{T}}\bm{U}_q\bm{U}^{\mathrm{T}}_q\bm{x} - \bm{x}^{\mathrm{T}}\bm{U}_q\bm{U}^{\mathrm{T}}_q\bm{x} \\ \nonumber && +  \mathbb{E}\left\{(\bm{x} + \bm{n})^{\mathrm{T}}\bm{U}_q\bm{U}^{\mathrm{T}}_q(\bm{x} + \bm{n})\right\}\\
\nonumber  &=&\!\!\! \bm{x}^{\mathrm{T}}\bm{U}_A\bm{U}^{\mathrm{T}}_A\bm{x} - \bm{x}^{\mathrm{T}}\bm{U}_q\bm{U}^{\mathrm{T}}_q\bm{x} - \bm{x}^{\mathrm{T}}\bm{U}_q\bm{U}^{\mathrm{T}}_q\bm{x} \\ \nonumber &&  + \mathbb{E}\left\{\bm{x}^{\mathrm{T}} \bm{U}_q\bm{U}^{\mathrm{T}}_q\bm{x}  +  \bm{x}^{\mathrm{T}} \bm{U}_q\bm{U}^{\mathrm{T}}_q\bm{n} + \bm{n}^{\mathrm{T}} \bm{U}_q\bm{U}^{\mathrm{T}}_q\bm{x}\right. \\ \nonumber && \left.  + \bm{n}^{\mathrm{T}} \bm{U}_q\bm{U}^{\mathrm{T}}_q\bm{n}\right\} \\
\nonumber  &=&\!\!\!  \bm{x}^{\mathrm{T}}\bm{U}_A\bm{U}^{\mathrm{T}}_A\bm{x} - \bm{x}^{\mathrm{T}} \bm{U}_q\bm{U}^{\mathrm{T}}_q\bm{x} + \mathbb{E}\left\{\bm{n}^{\mathrm{T}} \bm{U}_q\bm{U}^{\mathrm{T}}_q\bm{n}\right\} \\
\nonumber  &=&\!\!\!   \sum_{j = 1}^{p+1}\norm{\bm{u}^{\mathrm{T}}_{(j)}\bm{x}}^2 - \sum_{j = 1}^{q}\norm{\bm{u}^{\mathrm{T}}_{(j)}\bm{x}}^2 + q \sigma^2 \\
&=&\!\!\!   \sum_{j = q+1}^{p+1}\norm{\bm{u}^{\mathrm{T}}_{(j)}\bm{x}}^2 + q \sigma^2.
\label{eq:mse_q}
\end{eqnarray}
The rank reduction procedure will reduce the variance of the TLS estimator of $\bm{x}$ whenever 
\begin{eqnarray}
\nonumber [\norm{\hat{\bm{H}} \bm{\theta} - \bm{H} \bm{\theta}}^2 + \sigma^2 (1 + \bm{\theta}^{\mathrm{T}}\bm{\theta})p]  > \sum_{j = q+1}^{p+1}\norm{\bm{u}^{\mathrm{T}}_{(j)}\bm{x}}^2 + q \sigma^2,
\label{eq:reduce__tls_when1}
\end{eqnarray}
which can be written as 
\begin{align}
\nonumber p \sigma^2   &> \frac{1}{[1 + \bm{\theta}^{\mathrm{T}}\bm{\theta}]}\left[\sum_{j = q+1}^{p+1}\norm{\bm{u}^{\mathrm{T}}_{(j)}\bm{x}}^2 + q \sigma^2 - \norm{\hat{\bm{H}} \bm{\theta} - \bm{H} \bm{\theta}}^2\right]  \\
\nonumber  &= \frac{1}{[1 + \bm{\theta}^{\mathrm{T}}\bm{\theta}]}\left[\sum_{j = q+1}^{p+1}\norm{\bm{u}^{\mathrm{T}}_{(j)}\bm{x}}^2 + q \sigma^2 - \norm{\hat{\bm{H}} \bm{\theta} - \bm{x}}^2\right] \\
&> \frac{1}{[1 + \bm{\theta}^{\mathrm{T}}\bm{\theta}]}\left[\sum_{j = q+1}^{p+1}\norm{\bm{u}^{\mathrm{T}}_{(j)}\bm{x}}^2 + q \sigma^2 \right].
\label{eq:reduce__tls_when2}
\end{align}
The optimum choice of $q$ can be obtained by solving the following optimization setup:
\begin{eqnarray}
\nonumber q^{\ast} &=& \arg\min_{q} \frac{1}{[1 + \bm{\theta}^{\mathrm{T}}\bm{\theta}]}\Biggl[\sum_{j = q+1}^{p+1}\norm{\bm{u}^{\mathrm{T}}_{(j)}\bm{y}}^2 + \Biggr. \\ && \Biggl. \sigma^2[1 + \bm{\theta}^{\mathrm{T}}\bm{\theta}](2q + p)\Biggr].
\label{eq:optimum_q}
\end{eqnarray}
However, as can be seen in \eqref{eq:optimum_q}, $q^{\ast}$ depends on the parameter $\bm{\theta}$ which is unknown. Thus, the hypothesis of the existence of the low rank estimator is false. In general, it is not possible to construct a low rank estimator of the form \eqref{eq:reduced_tls} that exploits the bias-variance tradeoff for the TLS formulation. This completes the proof of \thrmref{thrm:reduced_tls}.
\end{proof}
The main essence of \thrmref{thrm:reduced_tls} is that, the estimator $\hat{\bm{\theta}}_{\text{TLS}}$ given by \eqref{eq:tls_min} does not permit rank reduction to enable the bias-variance tradeoff for the TLS formulation specified in \eqref{eq:tls}. However, there are certain applications where constraints on the norm of the parameter to be estimated plays a critical role in the estimation problem; see, for example, \cite{Douglas2000} and \cite{Gotoh2011}. For instance, when $\bm{\theta}^{\mathrm{T}}\bm{\theta} \leq C$, where $C > 0$ is some constant, equation \eqref{eq:optimum_q} can be solved for $q^{\ast}$. Once the optimal $q = q^{\ast}$ is obtained, similar to the reduced-rank LS estimator, the eigenvectors of $\bm{U}_s$ should be ordered such that 
$\norm{\bm{u}^{\mathrm{T}}_{(1)}\bm{y}}^2 \geq \cdots \geq \norm{\bm{u}^{\mathrm{T}}_{(q)}\bm{y}}^2 \geq \cdots \geq \norm{\bm{u}^{\mathrm{T}}_{(p+1)}\bm{y}}^2$,
and the dominant $q$ eigenvectors can be used to construct the rank$-q$ projector $\bm{U}_q\bm{U}^{\mathrm{T}}_q$ in order to derive the low-rank estimator $\hat{\bm{x}}_{q, {\text{TLS}}} \triangleq \bm{U}_q\bm{U}^{\mathrm{T}}_q\bm{y}$.

\section{Remarks}\label{sec:conclusion}
The reduced-rank method exploits the bias-variance tradeoff to yield an estimator having a lower bias-plus-variance compared to its unbiased counterpart. It also results in an ordering principle to select the most informative eigenvectors of the measurement matrix and conveniently discard the least informative ones. For the TLS formulation, we showed that rank reduction is, in general, infeasible owing to the structure of the problem. In the LS problem, the unknown parameter appears only in the mean of the observations [see \eqref{eq:linear_model}], thereby providing scope to exploit the bias-variance tradeoff. On the other hand, in the TLS formulation, the unknown parameter appears in both the mean and variance of the observations [note that, $\bm{y} \sim \mathcal{N}(\bm{H}\bm{\theta}, \sigma^2 [1 + \bm{\theta}^{\mathrm{T}}\bm{\theta}]\bm{I}_N)$], thus the notion of bias-variance tradeoff may not hold as straightforwardly as it does for the LS problem. Lastly, in the LS problem, the optimum $r^{\ast}$ can easily be obtained by solving \eqref{eq:optimum_r_mod} numerically. In case of TLS, attempting to obtain $q^{\ast}$ by solving the optimization setup in \eqref{eq:optimum_q} numerically is futile because of the unknown parameter $\bm\theta$ in the expression.


\section*{Acknowledgement}
This work was supported by the U.S. NSF under grant number NSF CNS$-1723856$. The first author thanks M. Fau{\ss} and A. Dytso for their comments on an earlier draft of the paper. 

\bibliographystyle{IEEEtran}
\bibliography{IEEEabrv,reducedtls}
\raggedbottom

\end{document}